\documentclass[12pt,a4paper,twoside]{article}      

\newcommand{\mechyear}{2006}
\newcommand{\mechvol}{36}

\usepackage{mathtext}
\usepackage[cp1251]{inputenc}
\usepackage[T2A]{fontenc}
\usepackage[russian]{babel}

\usepackage[dvips]{graphicx}
\usepackage{amsmath}
\usepackage{amssymb}
\usepackage{amsxtra}
\usepackage{latexsym}
\usepackage{ifthen}
\usepackage[centerlast,footnotesize,sl]{caption2}

\usepackage{mtt32a}

\newenvironment{proposition}[1][s]
{\vskip1mm \ifthenelse{\equal{#1}{s}}{{\small \sc
\langu{Proposition. }{Предложение. }}} {{\small \sc
\langu{Proposition #1.}{Предложение #1.}}}}{}

\begin{document}

\udk{531.38}

\title{ОБЛАСТИ СУЩЕСТВОВАНИЯ КРИТИЧЕСКИХ ДВИЖЕНИЙ\\
ОБОБЩЕННОГО ВОЛЧКА КОВАЛЕВСКОЙ\\И БИФУРКАЦИОННЫЕ ДИАГРАММЫ}
      {Области существования критических движений обобщенного волчка Ковалевской}

\author{М.П.~Харламов}

\date{11.04.06}

\address{Волгоградская академия государственной службы, Волгоград}

\email{mharlamov@vags.ru}

\maketitle

\begin{abstract}
Работа завершает цикл исследований бифуркационных диаграмм
гамильтоновой системы с тремя степенями свободы, описывающей
движение осесимметричного твердого тела с условиями типа Ковалевской
в двойном силовом поле. Получены явные неравенства, определяющие
множество критических значений первых интегралов на поверхностях,
несущих бифуркационную диаграмму (Харламов М.П., МТТ, 2004,
вып.~34). Выполнено построение всех диаграмм на изоэнергетических
уровнях, тип которых устойчив относительно малых изменений
физических параметров и постоянной интеграла энергии.

\end{abstract}

\Section{Предварительные сведения} Задача о движении волчка
Ковалевской в двойном силовом поле описывается системой уравнений
\begin{equation}
\begin{array}{c}
2\omega _1^ {\boldsymbol \cdot}   = \omega _2 \omega _3  + \beta _3
,\qquad 2\omega _2^ {\boldsymbol \cdot}   =  - \omega _1 \omega _3 -
\alpha _3 ,\qquad
\omega _3^ {\boldsymbol \cdot}   = \alpha _2  - \beta _1 , \\[1.5mm]
\alpha _1^ {\boldsymbol \cdot}   = \alpha _2 \omega _3  - \alpha_3
\omega_2,\qquad \beta _1^ {\boldsymbol \cdot}   = \beta _2 \omega _3
-
\beta_3 \omega_2  \qquad (123). \\
\end{array}
\end{equation}
Здесь и в дальнейшем используются соглашения и обозначения работы
[1]. Символ (123) означает, что невыписанные уравнения второй группы
получаются из приведенных циклической перестановкой индексов. Как
показано в [1], без ограничения общности можно полагать
характеристические векторы силовых полей
$\boldsymbol{\alpha},\boldsymbol{\beta}$ взаимно ортогональными и
ограничить систему (1) на многообразие $P^6$, заданное в
пространстве ${\bf{R}}^9 ({\boldsymbol{\omega }},{\boldsymbol{\alpha
}},{\boldsymbol{\beta }})$ геометрическими интегралами
\begin{equation}
\alpha _1^2  + \alpha _2^2  + \alpha _3^2  = a^2 ,\quad \beta _1^2 +
\beta _2^2  + \beta _3^2  = b^2 ,\quad \alpha _1 \beta _1  + \alpha
_2 \beta _2  + \alpha _3 \beta _3  = 0.
\end{equation}
Полагаем для определенности $a \geqslant b \geqslant 0$. Если $a
> b > 0$, то система (1), (2) гамильтонова с тремя степенями свободы,
не имеет явных групп симметрий и не сводится к семейству
гамильтоновых систем с двумя степенями свободы. Полную
интегрируемость по Лиувиллю обеспечивают интегралы в инволюции [2,
3]
\begin{equation}
\begin{array}{l}
H = \omega _1^2  + \omega _2^2  + {1 \over 2}\omega _3^2  - (\alpha
_1  + \beta _2 ), \\[2mm]
K = (\omega _1^2  - \omega _2^2  + \alpha
_1 - \beta _2 )^2  +
(2\omega _1 \omega _2  + \alpha _2  + \beta _1 )^2 , \\[2mm]
G = (\alpha _1 \omega _1  + \alpha _2 \omega _2  + {1 \over 2}\alpha
_3 \omega _3 )^2  + (\beta _1 \omega _1  + \beta _2 \omega _2  + {1
\over 2}\beta _3 \omega _3 )^2  +  \\
\qquad {} + \omega _3 (\gamma _1 \omega _1  + \gamma _2 \omega _2 +
{1 \over 2}\gamma _3 \omega _3 ) - \alpha _1 b^2  - \beta _2 a^2
\end{array}
\end{equation}
($\gamma_i$ -- компоненты в подвижных осях вектора
${\boldsymbol{\alpha }} \times {\boldsymbol{\beta }}$).

В работах [1, 4] найдено множество критических точек $c(J)$
интегрального отображения $ J=H \times K \times G:P^6 \to
{\bf{R}}^3$ и указаны уравнения поверхностей в пространстве ${\bf
R}^3(h,k,g)$, несущих в себе бифуркационную диаграмму $\Sigma (J)$
отображения $J$. Оказалось, что $c(J)$ представимо в виде
объединения трех инвариантных подмножеств $\mathfrak{M},
\mathfrak{N}, \mathfrak{O}$, каждое из которых является почти всюду
гладким четырехмерным многообразием и в гладкой части задается в
$P^6 $ системой двух инвариантных соотношений. Первое множество
$\mathfrak{M}$ было ранее указано в [2] как многообразие частной
интегрируемости системы (1). Оно совпадает с нулевым уровнем
интеграла~$K$. Следовательно, этот класс движений обобщает 1-й класс
Аппельрота~[5] движений волчка Ковалевской в поле силы тяжести~[6] .
Фазовая топология этого случая (без учета интегрируемости системы в
целом) изучена Д.Б.~Зотьевым~[7]. Множество $\mathfrak{N}$ найдено в
[8] как обобщение семейства особо замечательных движений 2-го и 3-го
классов Аппельрота. Движения на $\mathfrak{N}$ полностью исследованы
в [9, 10]. Свойства системы, индуцированной на множестве
$\mathfrak{O}$ (обобщение 4-го класса Аппельрота), изучались в
работе~[11]. Относительно бифуркационной диаграммы в работах~[1, 4]
показано, что

1) $J(\mathfrak{M})\subset \Sigma_1$, где $\Sigma_1 = \{k=0\}$;

2) $J(\mathfrak{N})\subset \Sigma_2$, где $\Sigma_2 =
\{(2g-p^2h)^2-r^4k=0\}$;

3) $J(\mathfrak{O})\subset (\Sigma_3 \cup \Sigma_4)$, где $\Sigma_3$
-- пара прямых
\begin{equation}
g = \pm abh,\qquad k = (a \mp b)^2,
\end{equation}
а $\Sigma_4$ -- поверхность кратных корней многочлена ${\psi (s) =
s^2 (s - h)^2  + (p^2  - k)s^2  - 2gs + \gamma^2}$, который при
исчезновении второго силового поля обращается в резольвенту Эйлера
второго многочлена Ковалевской. Здесь введены положительные
параметры $p,\, r,\, \gamma$: \; $p^2 = a^2+b^2,\, r^2=a^2-b^2,\,
\gamma=ab$. Уравнение дискриминантной поверхности $\Sigma_4$ в форме
${D(h,k,g)=0}$ неразрешимо явно относительно какой-либо из
переменных $h,k,g$. В дальнейшем удобно считать $h$ и $s$
независимыми параметрами на $\Sigma_4$. Тогда система $\psi(s)
=0,\,\psi'(s)=0$ дает зависимости $g,k$ от $s,h$ на $\Sigma_4$:
\begin{equation}
g = s^2(h-s)+\frac {\gamma^2}{s},\qquad k=3s^2-4h s+h^2+p^2-\frac
{\gamma^2}{s^2}.
\end{equation}

В работе [4] сформулирована задача полного описания бифуркационной
диаграммы $\Sigma=\Sigma (J)$ путем нахождения условий на постоянные
первых интегралов, определяющих $\Sigma$ в объединении $\widetilde
\Sigma = \Sigma_1 \cup \Sigma_2 \cup \Sigma_3 \cup \Sigma_4$. Точки
множества $\Sigma \cap \Sigma_i$ будем для краткости называть
допустимыми точками листа $\Sigma_i$. Предложено искать условия
допустимости на сечениях бифуркационной диаграммы плоскостями,
параллельными плоскости $Ogk$. Тем самым значение энергии $h$
выбирается в качестве параметра, что соответствует исследованию
бифуркационных диаграмм $\Sigma^h$ пары интегралов $G,K$ на
изоэнергетических уровнях $\{\zeta \in P^6: H(\zeta) = h\}$. В
работах [12, 4] намечен следующий подход к решению поставленной
задачи. Представим множество критических точек $c(J)$ как
объединение непересекающихся множеств $\mathfrak{C}^{(i)}=\{\zeta
\in c(J) : \mathop{\rm rank}\nolimits J(\zeta) = i\}$. Тогда
$\Sigma$ можно рассматривать как двумерный клеточный комплекс;
объединение $\mathfrak{S}^{(i)}$ клеток размерности $i\,(0 \leqslant
i\leqslant 2)$ совпадает с $J(\mathfrak{C}^{(i)})$. Очевидно, что
$\partial \mathfrak{S}^{(2)} \subset \mathfrak{S}^{(0)} \cup
\mathfrak{S}^{(1)}$, $\partial \mathfrak{S}^{(1)} \subset
\mathfrak{S}^{(0)}$, т.е. фактические граничные точки допустимых
множеств на листах $\Sigma_i$ (точки $\Sigma_i$, через которые
бифуркационная диаграмма не может быть гладко продолжена в
рассматриваемом листе или на другой лист, касающийся его в данной
точке) содержатся в образе множества $\{\zeta \in c(J) : \mathop{\rm
rank}\nolimits J(\zeta) < 2\}$. Поэтому для того чтобы получить
граничные условия для $\Sigma \subset \widetilde \Sigma$, необходимо
указать все траектории в фазовом пространстве задачи, на которых
ранг интегрального отображения падает более, чем на единицу, т.е.
$\mathop{\rm rank}\nolimits J = 0$ или $\mathop{\rm rank}\nolimits J
= 1$. Ниже найдены все такие движения и, путем рассмотрения их
образов на листах поверхности $\widetilde \Sigma$, получены явные
неравенства, определяющие бифуркационную диаграмму. Построены все
сечения $\Sigma^h$ в зависимости от постоянной интеграла энергии и
физических параметров задачи.

\Section {Случаи сильного падения ранга интегрального отображения}
Отметим некоторые факты, имеющие место для произвольного
осесимметричного тела, вращающегося в двух независимых постоянных
полях с центрами приложения в экваториальной плоскости.

\begin{proposition}[1][13] При условии $a\ne b \ne 0$
уравнения движения тела имеют на $P^6$ ровно четыре положения
равновесия $c_j\; (j = 0,1,2,3)$:
$$
\begin{array}{lllll}
c_0: & {\boldsymbol \omega} = 0, & {\boldsymbol \alpha} = {\bf e}_1
, & {\boldsymbol \beta} = {\bf e}_2 ; & H(c_0 ) =  - a - b;\\
c_1: & {\boldsymbol \omega} = 0, & {\boldsymbol \alpha} = {\bf e}_1, & {\boldsymbol \beta} =  - {\bf e}_2 ; & H(c_1 ) =  - a + b;\\
c_2: & {\boldsymbol \omega} = 0, & {\boldsymbol \alpha} =  - {\bf e}_1, & {\boldsymbol \beta} = {\bf e}_2 ; & H(c_2 ) = a - b;\\
c_3: & {\boldsymbol \omega} = 0, & {\boldsymbol \alpha} =  - {\bf
e}_1, & {\boldsymbol \beta } =  - {\bf e}_2 ; & H(c_3 ) = a + b.
\end{array}
$$
При этом индекс Морса энергии $H$ в точке $c_j$ равен $j$. В
частности, $c_0$ -- точка глобального минимума $H$.
\end{proposition}

Отсюда следует также, что все изоэнергетические многообразия связны.
Это свойство оказывается весьма полезным при глобальном исследовании
областей существования произвольных (не обязательно критических
движений), так как тогда любая непрерывная функция при фиксированной
энергии принимает все значения от наименьшего до наибольшего, а эти
последние фактически устанавливаются ниже.

\begin{proposition}[2][1] Уравнения движения тела имеют на $P^6$ следующие
семейства решений маятникового типа:
\begin{gather}
{\boldsymbol \alpha } \equiv  \pm a{\bf e}_1, \; {\boldsymbol \beta
} = b({\bf e}_2 \cos \theta - {\bf e}_3 \sin \theta ), \;
{\boldsymbol \omega } = \theta ^ {\boldsymbol \cdot} {\bf e}_1 , \;
2\theta ^{ {\boldsymbol \cdot}  {\boldsymbol \cdot} } =  - b\sin
\theta; \\[1.5mm]
{\boldsymbol \alpha } = a({\bf e}_1 \cos \theta + {\bf e}_3 \sin
\theta ), \; {\boldsymbol \beta } \equiv  \pm b{\bf e}_2 , \;
{\boldsymbol \omega } = \theta ^ {\boldsymbol \cdot}  {\bf e}_2 , \;
2\theta ^{ {\boldsymbol \cdot}  {\boldsymbol \cdot} }  =  - a\sin
\theta; \\[1.5mm]
{\boldsymbol{\alpha }} = a({\bf{e}}_1 \cos \theta  - {\bf{e}}_2 \sin
\theta ),\;{\boldsymbol{\beta }} =  \pm b({\bf{e}}_1 \sin \theta  +
{\bf{e}}_2 \cos \theta ), \; {\boldsymbol{\omega }} = \theta ^
{\boldsymbol \cdot}  {\bf{e}}_3 ,\; \theta ^{ {\boldsymbol \cdot}
{\boldsymbol \cdot} }  =  - (a \pm b)\sin \theta.
\end{gather}
\end{proposition}

Вернемся к обобщенному волчку Ковалевской. Напомним, что
$\mathfrak{M}=\{K=0\} \subset P^6$ и $J(\mathfrak{M})= \Sigma_1 \cap
\Sigma$. Известно~[7], что
\begin{equation}
(p^2H-2G) |_\mathfrak{M} \equiv \frac{1}{2}F^2,
\end{equation}
где $F:\mathfrak{M} \to {\bf R}$ -- частный интеграл
О.И.~Богоявленского. Константу этого интеграла обозначим через $f$.
Обозначим $H^{(1)}=H |_\mathfrak{M} $ и рассмотрим интегральное
отображение $J^{(1)}=H^{(1)}\times F:\,\mathfrak{M} \to {\bf R}^2$.

\begin{proposition}[3][7] Бифуркационная диаграмма отображения
$J^{(1)}$ есть множество решений уравнения
\begin{equation}
[27 f^4
-9(h^2-2p^2)(2hf^2+r^4)+2(h^2-2p^2)^3]^2-4[(h^2-2p^2)^2-3(2hf^2+r^4)]^3
= 0
\end{equation}
в области $\{h \geqslant -2b\}$. Множество
допустимых значений $J^{(1)}$ определяется неравенствами $ |f|
\leqslant f_0(h)$, $h \geqslant -2b$, где через $f_0(h)$ обозначен
наибольший положительный корень уравнения $(10)$ относительно $f$.
\end{proposition}

Вообще говоря, из зависимости функций $H^{(1)},F$ на подмногообразии
$\mathfrak{M} \subset P^6$, даже при том, что на $\mathfrak{M}$
всюду $dK \equiv 0$, не следует, что в этих же точках $\mathop{\rm
rank}\nolimits J = 1$. Однако ниже будет показано, что это
действительно так. В следующих утверждениях содержится полное
описание множества $\mathfrak{C}^{(0)} \cup \mathfrak{C}^{(1)}$ и
его образа в пространстве констант первых интегралов.

\begin{theorem}[1]
Ранг отображения $J$ равен нулю в точности в положениях равновесия
тела. Все четыре положения равновесия принадлежат $\mathfrak{N} \cap
\mathfrak{O}$. В частности, их $J$-образы принадлежат одновременно
$\Sigma_2,\Sigma_3$ и $\Sigma_4$. На плоскости $(s,h)$ им
соответствуют точки
\begin{equation}
\begin{array}{ll}
Q_{01} = (-a,-(a+b)), & Q_{02} = (-b,-(a+b));\\
Q_{11} = (-a,-(a-b)), & Q_{12} = (b,-(a-b));\\
Q_{21} = (a,a-b), & Q_{22} = (-b,a-b);\\
Q_{31} = (a,a+b), & Q_{32} = (b,a+b).
\end{array}
\end{equation}
\end{theorem}
\begin{proof}
Если ранг отображения $J$ равен нулю, то в такой точке обязательно
$dH=0$. Это имеет место лишь в положениях равновесия соответствующей
гамильтоновой системы, т.e. в неподвижных точках системы (1), (2).
Они описаны в предложении~1. Непосредственно проверяется, что в этих
точках $dG=0, dK=0$. Поэтому $\mathfrak{C}^{0}=\{c_0,c_1,c_2,c_3\}$.
Обозначим $P_j=J(c_j)$. Координаты $P_j$ вычисляются из (3)
$$
\begin{array}{lllll}
P_0: & h=-(a+b), & g=-ab(a+b), & k=(a-b)^2;\\
P_1: & h=-(a-b), & g=ab(a-b), & k=(a+b)^2;\\
P_2: & h=a-b, & g=-ab(a-b), & k=(a+b)^2;\\
P_3: & h=a+b, & g=ab(a+b), & k=(a-b)^2
\end{array}
$$
и, очевидно, удовлетворяют одному из соотношений (4) и уравнению
поверхности $\Sigma_2$. Подставим эти значения в (5) с тем, чтобы
проверить существование решения относительно $s$. Получим выражения
(11). Поэтому $P_j \in \Sigma_4$. Здесь следует отметить, что
отображение (5) плоскости $(s,h)$ в $\Sigma_4$ не взаимнооднозначно.
\end{proof}

\begin{theorem}[2]
Траектории, удовлетворяющие условиям $\mathop{\rm rank}\nolimits J
=1$ и $dK \ne 0$, исчерпываются движениями $(6)$, $(7)$.
Соответствующие значения первых интегралов заполняют кривые
$$
\begin{array}{llll}
\lambda_1: & g = a^2 h+a r^2 , & k=(h+2a)^2 & h \geqslant -(a+b);
\\
\lambda_2: & g = a^2 h-a r^2, & k=(h-2a)^2,& h \geqslant a-b;\\
\lambda_3: & g = b^2 h-b r^2, & k=(h+2b)^2,& h \geqslant -(a+b);
\\
\lambda_4: & g = b^2 h+b r^2, & k=(h-2b)^2,& h \geqslant -(a-b),
\end{array}
$$
принадлежащие $\Sigma_2 \cap \Sigma_4$. На плоскости $(s,h)$ эти
кривые изображаются лучами
$$
\begin{array}{llllll}
\lambda_1: & s=-a, & h \geqslant -(a+b); & \quad \lambda_2: & s=a, & h \geqslant a-b;\\
\lambda_3: & s=-b, & h \geqslant -(a+b); & \quad \lambda_4: & s=b, &
h \geqslant -(a-b).
\end{array}
$$
\end{theorem}
\begin{proof} Ненулевой ранг $J$ предполагает, что $dH \ne 0$. Поскольку $dK \ne
0$, должно выполняться условие $dH = \lambda dK$ с некоторой
постоянной $\lambda \ne 0$. Компонента с ${\partial
/\partial\omega_3}$ дает $\omega_3 \equiv 0$. Из системы (1)
вытекает, что такое тождество по $t$ может иметь место только для
маятниковых движений (6), (7). Значения $g,k,s$, а также пределы
изменения $h$ вычисляются непосредственно из выражений (3) и
уравнений (5).
\end{proof}

\begin{theorem}[3]
Траектории, удовлетворяющие условиям $\mathop{\rm rank}\nolimits J
=1, dK =0$ и $K \ne 0$, исчерпываются движениями $(8)$. Их $J$-образ
заполняет лучи
\begin{equation}
\begin{array}{lll}
g = abh,& k = (a - b)^2 ,& h \geqslant  - (a +
b),\\
g =  - abh,& k = (a + b)^2 ,& h \geqslant  - (a - b).
\end{array}
\end{equation}
Часть этого множества, лежащая в $\Sigma_4$, определяется следующими
отрезками кривых в плоскости $(s,h)$:
$$
\begin{array}{ll}
\mu_1: \displaystyle{h=s-\frac{ab}{s}}, \, s \in [-a,0); & \mu_2:
\displaystyle{h=s-\frac{ab}{s}}, \, s \in [b,+\infty);\\[2mm]
\mu_3:
\displaystyle{h=s+\frac{ab}{s}}, \, s \in [-a,-b]; & \mu_4:
\displaystyle{h=s+\frac{ab}{s}}, \, s \in (0,+\infty).
\end{array}
$$
\end{theorem}
\begin{proof}
В предположениях теоремы имеем $K \ne 0$, $K_{\omega_1}=
K_{\omega_2}=0$, откуда следует, что $\omega_1 = \omega_2 \equiv 0$.
Тогда в силу системы (1) приходим к траекториям (8) и выражениям
(4). Допустимые значения $h$ в (12) вычисляются подстановкой (8) в
выражение для $H$ из (3). Подстановка же значений $g,k$ из (4) в (5)
дает искомые зависимости $h$ от $s$, а допустимые значения $h$ в
(12) определяют и допустимые значения $s$.
\end{proof}

\begin{theorem}[4]
Траектории, удовлетворяющие условиям $\mathop{\rm rank}\nolimits J
=1$, $dK = 0$ и $K = 0$, исчерпываются критическими движениями
гамильтоновой системы с двумя степенями свободы, индуцированной на
многообразии $\mathfrak{M}$. Эти траектории также лежат в
~$\mathfrak{O}$ и порождают следующие отрезки кривых на плоскости
$(s,h)$:
\begin{equation}
\begin{array}{lll}
\delta_1: & h=2s+\phi(s), & s \in [-b,0);\\
\delta_2: & h=2s+\phi(s), & s \in (0,b];\\
\delta_3: & h=2s-\phi(s), & s \in [a,+\infty),
\end{array}
\end{equation}
где $
\phi(s)=\displaystyle{\sqrt{(s^2-a^2)(s^2-b^2)/{s^2}} \geqslant 0}$.
\end{theorem}
\begin{proof}
Достаточно предположить, что $K=0$. Тогда автоматически $dK=0$, $dH
\ne 0$. Следовательно $\mathop{\rm rank}\nolimits J \leqslant 2$, и
ранг равен единице в точности в точках линейной зависимости $dH,dG$,
принадлежащих $\mathfrak{M}=\{K=0\}$. Поскольку в таких точках
должны существовать (с точностью до линейного преобразования) две
равные нулю независимые нетривиальные линейные комбинации
$dH,dG,dK$, то бифуркационные поверхности $\Sigma_1, \Sigma_4$
должны в соответствующих точках пересекаться трансверсально. Полагая
в (5) $k=0$, получаем искомые зависимости $h$ от $s$. Области
изменения $s$ вычисляются с применением предложения~3.
\end{proof}

\Section{Описание допустимого множества} Рассмотрим поверхность
$\Sigma_1$. Несмотря на то что существование и расположение корней
уравнения (10) полностью исследовано аналитически в работе~[7],
получить явные выражения для наибольшего положительного корня
$f_0(h)$ оказывается невозможным. Заметим, что подстановки (13)
позволяют параметрически выразить корни (10) в виде
$(h_i(s),f_i(s))$. Отсюда, в частности, следует, что все критические
точки частного интегрального отображения предложения~3 одновременно
являются точками, в которых $\mathop{\rm rank}\nolimits J = 1$.
Корень $f_0(h)$ отвечает кривой $\delta_1$. На ней зависимость $h$
от $s$ монотонно возрастающая. Обозначим ее обращение через $s_1(h)$
($h \in [-2b,+\infty],\, s_1(h) \in [-b,0)$). Тогда из (9), (10)
найдем выражение
$$
g = g_1(h) = s^3+\frac{ab}{s}-s^2 \phi(s)|_{s=s_1(h)}, \quad h
\geqslant -2b.
$$
Из (9) имеем также $p^2h - 2g \geqslant 0$. Теперь из предложения~3
и теоремы~4 получаем следующий результат.
\begin{theorem}[5]
Множество $\Sigma_1 \cap \Sigma$ имеет вид $ k = 0,\, h \geqslant
-2b,\, g_1(h) \leqslant g \leqslant \frac{1}{2}p^2h$.
\end{theorem}

Необходимо отметить, что верхняя граница для $g$ не является образом
точек сильного падения ранга $J$, но она также и не принадлежит
$\partial \mathfrak{S}^{(2)}$, так как через эти точки имеется
гладкое продолжение $\Sigma_1 \cap \Sigma$ на поверхность
$\Sigma_2$. Наличие этой границы порождено лишь условным делением
$\widetilde \Sigma$ на $\Sigma_i$.

Обратимся к поверхности $\Sigma_2$. Как и в предыдущем случае, здесь
возникает некоторая естественная граница допустимой области, не
связанная с бифуркациями внутри $c(J)$, а объясняющаяся возможностью
гладкого продолжения $\Sigma_2 \cap \Sigma$ на поверхность
$\Sigma_4$. Множество точек касания поверхностей $\Sigma_2$ и
$\Sigma_4$ описывается уравнением ${2p^2 (p^2h-2g)^2 - 2h (p^2h
-2g)r^4 + r^8 =0}$, откуда вытекают зависимости
$$
\displaystyle{g_{\pm}(h)= \frac{1}{4p^2}[(2p^4-r^4)h \pm
r^4\sqrt{h^2-2p^2}].}
$$
Из предложения~1 следует, что $h \geqslant -(a+b)$. Поэтому
значениями первых интегралов является лишь та часть этого множества,
где $h\geqslant \sqrt{2}p$. При этом $g_{-}(h)< g_{+}(h)$, если $h
> \sqrt{2}p$. Как показано в [9], значения $g \in
(g_{-}(h),g_{+}(h))$ на $\Sigma_2$ недопустимы. Отметим еще особые
значения $h=(3a^2+b^2)/2a,\,(a^2+3b^2)/2b$, при которых кривые
$g_{\pm}$ касаются лучей -- проекций кривых $\lambda_2, \lambda_4$
на плоскость $(h,g)$. Из теорем~1,~2 получаем следующее описание
допустимого множества.
\begin{theorem}[6]
Множество $\Sigma_2 \cap \Sigma$ лежит в полупространстве $h
\geqslant -(a+b)$ и описывается следующей совокупностью систем
неравенств
$$
\left\{ {\begin{array}{l} b^2 h -b r^2\leqslant g \leqslant a^2 h +a
r^2\\
-(a+b)\leqslant h \leqslant \sqrt{2}p
\end{array}} \right.; \quad
\left\{ {\begin{array}{l} b^2 h -b r^2 \leqslant g \leqslant g_{-}(h)\\
h \geqslant \sqrt{2}p
\end{array}} \right.; \quad
\left\{ {\begin{array}{l} g_{+}(h) \leqslant g \leqslant a^2 h +a
r^2 \\
h \geqslant \sqrt{2}p
\end{array}} \right..
$$
\end{theorem}

Для достаточно простого множества $\Sigma_3$ полное описание
допустимых значений дают неравенства в (12).

На поверхности $\Sigma_4$ удобнее описывать допустимое множество в
терминах переменных $s,h$ с учетом выражений (5). На кривой
$\delta_1$ имеем уже отмеченную зависимость $s_1(h), \,h \geqslant
-2b$. На кривой $\delta_2$ имеем $h'(s) < 0$. Обозначим обратную
зависимость через $s_2(h), \, h \geqslant 2b$, так что при этом
$s_2(2b)= b$ и $s_2(h)\to 0$ при $h \to +\infty$. На кривой
$\delta_3$ уравнение $h'(s)=0$ имеет единственное решение $s_0 \in
(a,+\infty)$. Пусть $h_0$ -- соответствующее минимальное значение
$h$. Очевидно, $h_0 \in (a+b,2a)$. Обозначим через $s_3(h),s_4(h)$
зависимости, определенные на кривой $\delta_3$ для $s \in [a,s_0]$ и
$s \in [s_0,+\infty)$ соответственно.

Рассмотривая интервалы монотонности $s(h)$ на кривых
$\mu_1$~--~$\mu_4$, обозначим
$$
\displaystyle{s_5(h)= \frac{h - \sqrt{h^2-4ab}}{2},}\quad
\displaystyle{s_6(h)= \frac{h + \sqrt{h^2-4ab}}{2},}\quad
\displaystyle{s_7(h)= \frac{h + \sqrt{h^2+4ab}}{2}}.
$$

Суммируя утверждения теорем~1~--~4 в части, относящейся к значениям
$s,h$, приходим к следующему утверждению.
\begin{theorem}[7]
Допустимая область $\Sigma_4 \cap \Sigma$ полностью описывается
следующей совокупностью условий на плоскости $(s,h)$. Для
отрицательных значений $s$:
$$
\left\{ {\begin{array}{l} -(a+b)\leqslant h \leqslant -2 \sqrt{ab}\\
s \in [-a,s_5(h)] \cup [s_6(h),-b]
\end{array}} \right.; \quad
\left\{ {\begin{array}{l} -2 \sqrt{ab}\leqslant h \leqslant -2b\\
s \in [-a,-b]
\end{array}} \right.; \quad
\left\{ {\begin{array}{l} h >  -2b\\
s \in [-a,s_1(h)]
\end{array}} \right..
$$

Для положительных значений $s$:
$$
\begin{array}{ll}
\left\{ {\begin{array}{l} -a+b\leqslant h \leqslant 2b\\
s \in [b,s_7(h)]\end{array}} \right.; &
\left\{ {\begin{array}{l} 2b \leqslant h \leqslant h_0\\
s \in [s_2(h),s_7(h)]
\end{array}} \right.; \\[4mm]
\left\{ {\begin{array}{l} h_0 \leqslant h \leqslant 2a\\
s \in [s_2(h),s_3(h)]\cup[s_4(h),s_7(h)]
\end{array}} \right.; &
\left\{ {\begin{array}{l} h > 2a\\
s \in [s_2(h),a]\cup[s_4(h),s_7(h)]
\end{array}} \right..
\end{array}
$$
\end{theorem}

Теоремы~5~--~7 дают полную информацию для построения средствами
компьютерной графики сечений бифуркационной диаграммы $\Sigma (J)$
плоскостями постоянной энергии.

\Section{Диаграммы на изоэнергетических уровнях} Условия на
постоянные первых интегралов в теоремах~1-7 записаны в терминах
конкретных значений постоянной энергии $h$ или промежутков ее
изменения и содержат параметры $a, b$. Равенства (теорема~1) и
граничные значения $h$ в неравенствах (теоремы~2-7) дают те значения
энергии, при переходе через которые в диаграмме $\Sigma^h$
происходят структурные изменения. Выписав эти значения, получаем
уравнения разделяющих поверхностей в пространстве ${\bf
R}^3(h,a,b)$. Кроме того, разделяющие поверхности возникают при
записи условий существования точек касания листов $\Sigma_i$ в
допустимой области. Всего таких поверхностей оказывается 13.
Обозначим их через $Q_1 - Q_{13}$:
\begin{equation}
\begin{array}{l}
 Q_1-Q_4 :\, h =  \pm a \pm b;  \qquad  Q_5-Q_6:\, h = \pm 2\sqrt {ab};\\
 \displaystyle{Q_7 :\,h = \sqrt {2(a^2+b^2 )};\qquad Q_8 :\,h = \frac{1} {2a}(3a^2 + b^2 ); \qquad Q_9 :\, h = \frac{1} {{2 b}}(a^2 + 3b^2)};\\
 \displaystyle{ Q_{10} :\, \left\{
 \begin{array}{l}
 \displaystyle{h = s(3 - \frac {s^2}{a^2} ) - \frac {1}{a^2}\sqrt {(s^2  - a^2)^3 }}
 \\[3mm]
 \displaystyle{ b  = \frac {s}{a}\sqrt {s [s(3 - \frac{2s^2}{a^2}) - \frac{2}{a^2}\sqrt {(s^2 -
 a^2)^3 } ]}}
 \end{array} \right.,\quad s \in [a,\frac{2}{\sqrt 3 }a];}\\
 Q_{11} :\, h = 2a;  \qquad Q_{12}-Q_{13} :\, h = \pm 2b .  \\
\end{array}
\end{equation}

\begin{figure}[h]
\centering
\includegraphics[width=\textwidth,keepaspectratio]{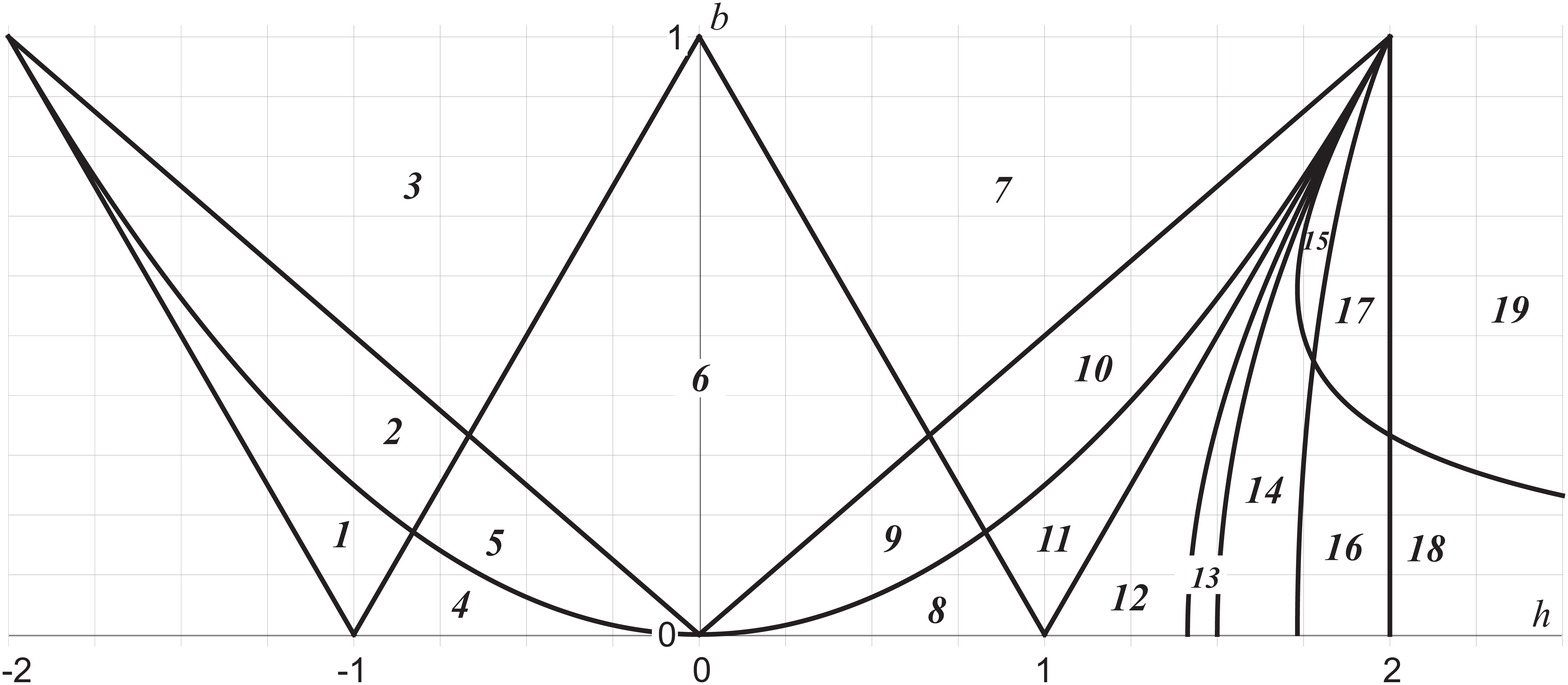}
\caption{Разделяющее множество в плоскости $(h,b)$.}
\end{figure}

\begin{figure}[ph]
\centering
\includegraphics[scale=1,clip]{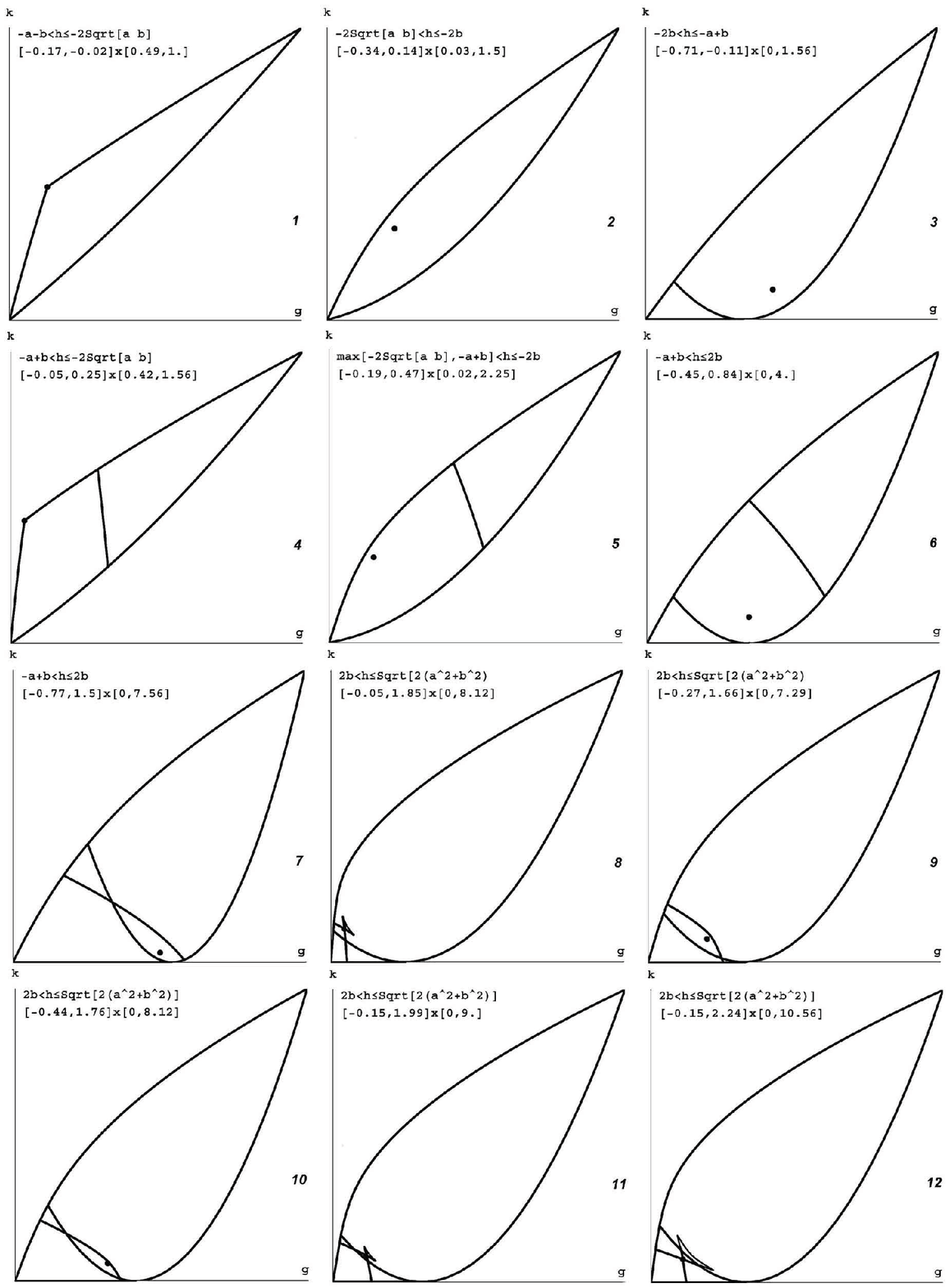}
\caption{Бифуркационные диаграммы для областей 1-12.}
\end{figure}

Заметим, что все сечения этих поверхностей плоскостями
фиксированного $a$ переводятся в сечение $a=1$ преобразованием
подобия (подстановка $h \to ah, b \to ab$, а для параметрически
заданной кривой $Q_{10}$ еще и $s \to as$). Это очевидное следствие
того факта, что величина $a$ может быть взята за единицу измерения
длин характеристических векторов силовых полей
$\boldsymbol{\alpha},\boldsymbol{\beta}$. Таким образом, без
ограничения общности полагаем $a=1$. Тогда уравнения (14) определяют
семейство разделяющих кривых в плоскости $(h,b)$. При этом $0
\leqslant b \leqslant 1$ и $h \geqslant -1-b$. Поэтому имеется ровно
19 областей, в которых диаграммы $\Sigma^h$ непусты и устойчивы
относительно малых изменений параметров (рис.~1), что совпадает с
результатом работы~[14], где уравнения разделяющих кривых получены
исходя из чисто геометрического рассмотрения особенностей плоских
сечений поверхности $\widetilde \Sigma$.

\begin{figure}[h]
\centering
\includegraphics[scale=1,clip]{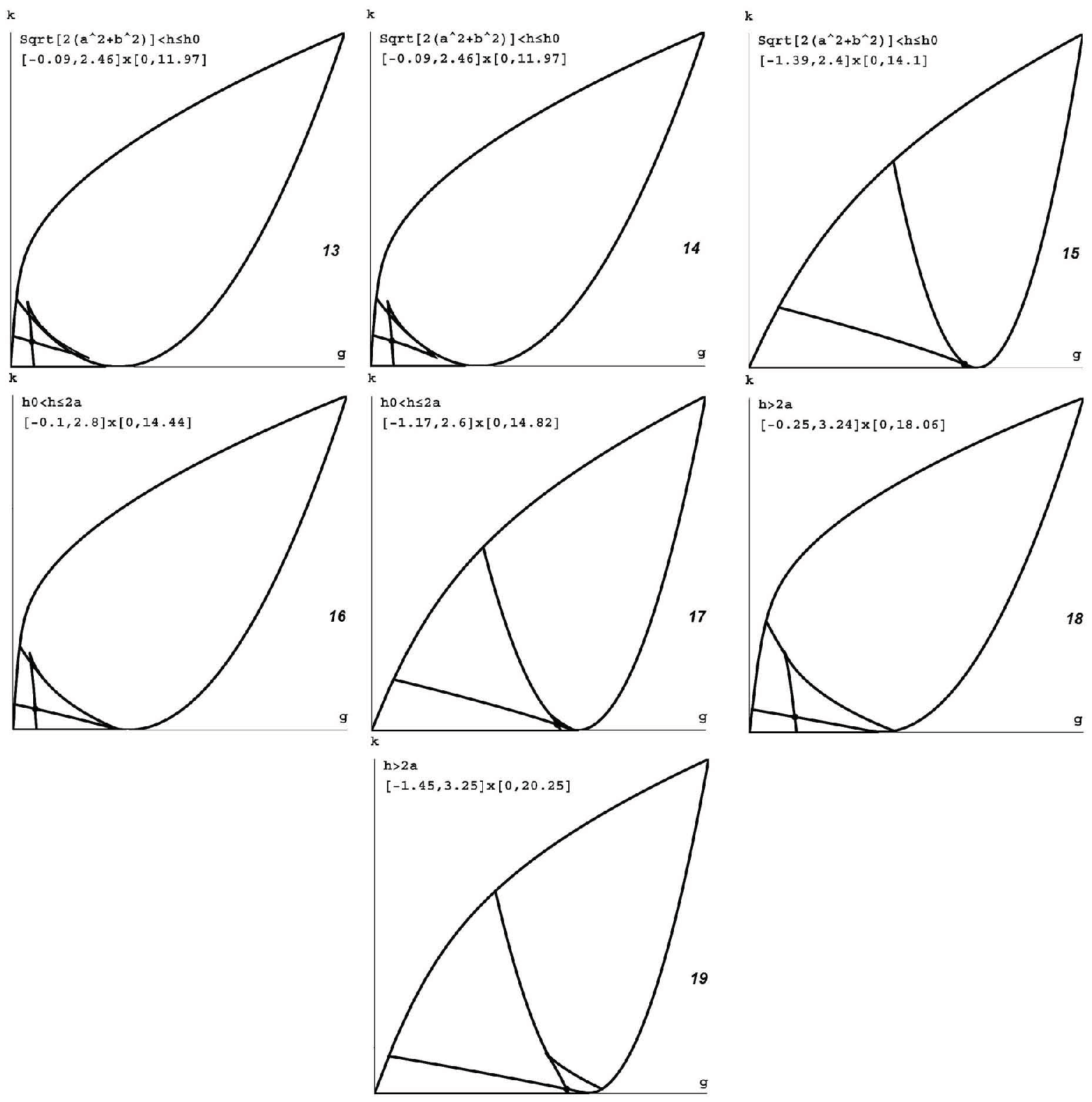}
\caption{Бифуркационные диаграммы для областей 12-19.}
\end{figure}

Полный набор типичных диаграмм показан на рис.~2,~3. Минимальные
прямоугольники, обрамляющие диаграммы $\Sigma^h$, масштабированы к
одному размеру изображения (фактические границы для $g$ и $k$
выведены на каждый рисунок). Отдельные элементы диаграмм для
областей 15-19 при показе всей диаграммы неразличимы. Однако легко
построить изображение необходимых участков. Кроме того, в работе
[14] описаны все явления, происходящие с особыми точками диаграмм на
разделяющем множестве.


\begin{thebibliography}{99}

\bibitem{bib1}
\emph {Харламов М.П.} Критическое множество и бифуркационная
диаграмма задачи о движении волчка Ковалевской в двойном поле //
Механика твердого тела. -- 2004. -- Вып.~34. -- С.~47--58.

\bibitem{bib2}
\emph {Богоявленский О.И.} Два интегрируемых случая динамики
твердого тела в силовом поле // Докл. АН СССР. -- 1984. -- {\bf
275}, №~6. -- С.~1359--1363.

\bibitem{bib3}
\emph {Рейман А.Г., Семенов-Тян-Шанский М.А.} Лаксово представление
со спектральным параметром для волчка Ковалевской и его обобщений //
Функц. анализ и его приложения. -- 1988. -- {\bf 22}, №~2.~--
С.~87--88.

\bibitem{bib4}
\emph {Kharlamov M.P.} Bifurcation diagrams of the Kowalevski top in
two constant fields // Регулярная и хаотическая динамика. -- 2005.
-- {\bf 10}, №~4. -- С.~381--398.

\bibitem{bib5}
\emph {Аппельрот Г.Г.} Не вполне симметричные тяжелые гироскопы
// Движение твердого тела вокруг неподвижной точки. -- 1940. -- М.;~Л.:
Изд-во АН СССР. -- С.~61--156.

\bibitem{bib6}
\emph {Ковалевская С.В.} Задача о вращении твердого тела около
неподвижной точки // Движение твердого тела вокруг неподвижной
точки. -- 1940. -- М.;~Л.: Изд-во АН СССР. -- С.~11--60.

\bibitem{bib7}
\emph {Зотьев Д.Б.} Фазовая топология 1-го класса Аппельрота волчка
Ковалевской в магнитном поле // Фундаментальная и прикладная
математика. -- 2006.  -- {\bf 12}, № 1. -- С.~95--128.

\bibitem{bib8}
\emph {Харламов М.П.} Один класс решений с двумя инвариантными
соотношениями задачи о движении волчка Ковалевской в двойном
постоянном поле // Механика твердого тела. -- 2002. -- Вып.~32. --
С.~32--38.

\bibitem{bib9}
\emph {Харламов М.П., Савушкин А.Ю, Шведов Е.Г.} Бифуркационное
множество в одной задаче о движении обобщенного волчка Ковалевской
// Механика твердого тела. -- 2003. -- Вып.~33. -- С.~10--19.

\bibitem{bib10}
\emph {Харламов М.П., Савушкин А.Ю.} Разделение переменных и
интегральные многообразия в одной частной задаче о движении
обобщенного волчка Ковалевской // Укр. матем. вестник. -- 2004. --
{\bf 1}, №~4. -- С.~548--565.

\bibitem{bib11}
\emph {Харламов М.П.} Бифуркационная диаграмма обобщения 4-го класса
Аппельрота // Механика твердого тела. -- 2005.~-- Вып.~35.~--
С.~38--48.

\bibitem{bib12}
\emph {Харламов М.П.} Общий подход к исследованию особых движений
обобщенного волчка Ковалевской~// Анн. докл. Пятого межд. Симпозиума
по классической и небесной механике. -- ВЦ РАН. -- 2004.~--
С.~207--208.

\bibitem{bib13}
\emph {Kharlamov M.P., Zotev D.B.} Non-degenerate energy surfaces of
rigid body in two constant fields // Регулярная и хаотическая
динамика. -- 2005. -- {\bf 10}, №~1. -- С.~15--19.

\bibitem{bib14}
\emph {Харламов М.П., Шведов Е.Г.} Бифуркационные диаграммы на
изоэнергетических уровнях волчка Ковалевской в двойном поле
// Механика твердого тела. -- 2004. -- Вып.~34. -- С.~59--65.

\end{thebibliography}
\end{document}